\documentclass[12pt, onecolumn,draftclsnofoot]{IEEEtran}

\usepackage{amsthm}
\usepackage{amsfonts}
\usepackage{amssymb}
\usepackage{mathrsfs}
\usepackage{mathtools}
\usepackage{float}
\usepackage[pdftex]{graphicx}
\usepackage{amsmath}
\usepackage{color}
\usepackage{multirow,multicol}
\usepackage{graphicx}
\usepackage{times}
\usepackage{textcomp}
\usepackage{verbatim}
\usepackage[table]{xcolor}
\usepackage{balance}
\usepackage{lipsum}
\usepackage[inline]{enumitem}
\usepackage{cuted}
\usepackage[caption=false,font=footnotesize]{subfig}
\usepackage{cite}
\usepackage{algpseudocode,algorithm}
\usepackage{hyperref}
\usepackage{tcolorbox}
\usepackage{setspace,epstopdf}
\usepackage[table]{xcolor}
\definecolor{color1}{RGB}{199,209,232}
\definecolor{color2}{RGB}{230,231,233}

\DeclareMathOperator*{\argmax}{argmax} 
\DeclareMathOperator*{\minimize}{minimize} 
\DeclareMathOperator*{\subjectto}{subject\hspace{3pt} to:\hspace{3pt}} 
\newtheorem{theorem}{Theorem}

\begin{document}
	
\title{\huge Near-Field Terahertz Communications: Model-Based and Model-Free Channel Estimation}
	
	
	\author{\IEEEauthorblockN{Ahmet M. Elbir, \textit{Senior Member, IEEE}, Wei Shi, \textit{Member, IEEE}, Anastasios K. Papazafeiropoulos, \textit{Senior Member, IEEE},
			Pandelis Kourtessis \textit{Member, IEEE}, and   Symeon Chatzinotas, \textit{Fellow, IEEE} }
		\thanks{This work was supported in part by the Natural Sciences and Engineering Research Council of Canada (NSERC) and Ericsson Canada.}
		\thanks{A. M. Elbir is with Interdisciplinary Centre for Security, Reliability and Trust, University of Luxembourg, Luxembourg (e-mail: ahmetmelbir@gmail.com).}
		\thanks{Wei Shi is with the	School of Information Technology, Carleton
			University, Ottawa, Canada (e-mail: wei.shi@carleton.ca ). }	
		\thanks{A. K. Papazafeiropoulos and P. Kourtessis are with the  the Communications and Intelligent Systems Research Group, University of Hertfordshire, Hatfield, U. K. (e-mail: tapapazaf@gmail.com,
			p.kourtessis@herts.ac.uk)} 
		
		\thanks{S. Chatzinotas is with Interdisciplinary Centre for Security, Reliability and Trust, University of Luxembourg, Luxembourg (e-mail: symeon.chatzinotas@uni.lu).}
		
	}
	
	\maketitle


	\begin{abstract}
		Terahertz (THz) band is expected to be one of the key enabling technologies of the sixth generation (6G) wireless networks because of its abundant available bandwidth and very narrow beam width.  Due to high frequency operations, electrically small array apertures are employed, and the signal wavefront becomes spherical in the near-field. Therefore, near-field signal model should be considered for channel acquisition in THz systems.  Unlike prior works which  mostly ignore the impact of near-field beam-split (NB) and consider either narrowband scenario or far-field models, this paper introduces both a model-based and a model-free techniques for wideband THz channel estimation in the presence of NB. The model-based approach is based on orthogonal matching pursuit (OMP) algorithm, for which we design an NB-aware dictionary. The key idea is to exploit the angular and range deviations due to the NB. We then employ the OMP algorithm, which accounts for the deviations thereby ipso facto mitigating the effect of NB. We further introduce a federated learning (FL)-based approach as a  model-free solution for channel estimation  in a multi-user scenario to achieve reduced complexity and training overhead.	Through numerical simulations, we demonstrate the effectiveness of the proposed channel estimation techniques for wideband THz systems in comparison with the existing state-of-the-art techniques.
	\end{abstract}

\section{Introduction}
\label{sec:Introduciton}
Terahertz (THz) band is expected to be a key component of the sixth generation  (6G) of wireless cellular networks because of its abundant available bandwidth. In particular, THz-empowered systems are envisioned to demonstrate revolutionary enhancement under high data rate ($>100\text{Gb/s}$), extremely low propagation latency ($<1\text{ms}$) and ultra reliability ($99.999\%$)~\cite{thz_Rappaport2019Jun,thz_Akyildiz2022May}.

Although demonstrating the aforementioned advantages, signal processing  in THz band faces several THz-specific challenges that should be taken into account accordingly. These challenges include, among others, severe path loss due to spreading loss and molecular absorption, extremely-sparse path model (see, e.g., Fig.~\ref{fig_PathLoss}), very short transmission distance and \textit{beam-split} (see the full list in~\cite{elbir2022Aug_THz_ISAC,ummimoTareqOverview}). In order to combat some of these challenges, e.g., path loss, analogues to massive multiple-input multiple-output (MIMO) arrays in millimeter-wave (mm-Wave) systems~\cite{heath2016overview,elbir2022Nov_Beamforming_SPM}, ultra-massive MIMO  architectures are envisioned, wherein subcarrier-independent (SI) analog beamformers are employed. In wideband signal processing, the weights of the analog beamformers are subject to a single subcarrier frequency,  i.e., carrier frequency~\cite{alkhateeb2016frequencySelective,hybridBFAltMin}. Therefore, the directions of the generated beams at different subcarriers differentiate and point to different directions causing {beam-split} phenomenon (see, e.g., Fig.~\ref{fig_ArrayGain_BS})~\cite{thz_beamSplit,elbir2022Aug_THz_ISAC}.

\begin{figure}[h]
	\centering
	{\includegraphics[draft=false,width=.5\columnwidth]{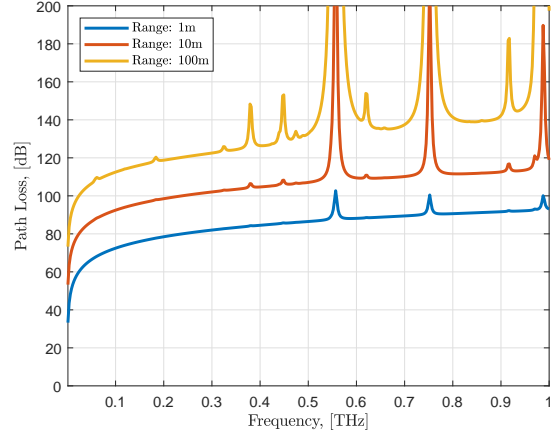} } 
	\caption{Path loss (in dB) due to molecular absorption for various transmission ranges.	}
	\label{fig_PathLoss}
\end{figure}

\subsection{Related Works}
While wideband mm-Wave channel estimation has been extensively studied in the literature~\cite{limitedFeedback_Alkhateeb2015Jul,heath2016overview,channelEstLargeArrays2,channelEstLargeArrays,channelEstimation1}, wideband THz channel estimation, on the other hand is relatively new~\cite{thz_Akyildiz2022May}. Specifically, the existing solutions are categorized into two classes, i.e., hardware-based techniques~\cite{dovelos_THz_CE_channelEstThz2} and algorithmic methods~\cite{thz_channelEst_beamsplitPatternDetection_L_Dai,elbir2022Jul_THz_CE_FL}. The first category of solutions consider employing time delayer (TD) networks together with phase shifters so that true-time-delay (TTD) of each subcarrier can be obtained for beamformer design. The TD networks are used to realize virtual SD analog beamformers so that the impact of beam-slit can be mitigated for hybrid beamforming~\cite{elbir2022_thz_beamforming_Unified_Elbir2022Sep}.  In particular, \cite{dovelos_THz_CE_channelEstThz2} devises a generalized simultaneous orthogonal matching pursuit (GSOMP) technique by exploiting the SD information collected via TD network hence achieves close to minimum mean-squared-error (MMSE) performance. However, these solutions require additional hardware, i.e., each phase shifter is connected to multiple TDs, each of which consumes approximately $100$ mW, which is more than that of a phase shifter ($40$ mW) in THz~\cite{elbir2022Aug_THz_ISAC}. The second category of solutions do not employ additional hardware components. Instead, advanced signal processing techniques have been proposed to compensate beam-split. Specifically, an OMP-based beam-split pattern detection (BSPD) approach was proposed in~\cite{thz_channelEst_beamsplitPatternDetection_L_Dai} for the recovery of the support pattern  among all subcarriers in the beamspace and construct one-to-one match between the physical and spatial (i.e., deviated due to beam-split in the beamspace) directions. Also,  \cite{spatialWidebandWang2018May} proposed an angular-delay rotation method, which suffers from coarse beam-split estimation and high training overhead due to the use of complete discrete Fourier transform (DFT) matrix. In \cite{elbir2022Jul_THz_CE_FL},  a beamspace support alignment (BSA) technique was introduced to align the deviated spatial beam directions among the subcarriers. Although both BSPD and BSA are based on OMP, the latter exhibits lower normalized MSE (NMSE) for THz channel estimation. Nevertheless, both methods suffer from inaccurate support detection and low precision for estimating the physical channel directions~\cite{elbir_THZ_CE_ArrayPerturbation_Elbir2022Aug}.

Besides the aforementioned model-based channel estimation techniques, model-free approaches, such as machine learning (ML), have also been suggested for THz channel estimation~\cite{thzCE_kernel,thzCE_CNN}. For instance, ML-based learning models such as  deep convolutional neural network (DCNN)~\cite{thzCE_CNN}, generative adversarial network (GAN)~\cite{thz_CE_GAN_Balevi2021Jan} and  deep kernel learning (DKL)~\cite{thzCE_kernel}, have been proposed to lower the  complexity involved during the channel inference as well as the complexity arising from the usage of the ultra-massive antenna elements. 
However, the works~\cite{thzCE_kernel,thzCE_CNN} only consider the narrowband THz systems, which do not exploit wideband scenario, which is the main reason to climb up to the THz-band to achieve better communication performance. In~\cite{thz_CE_GAN_Balevi2021Jan}, the wideband scenario is considered but the effect of beam-split is ignored. In~\cite{elbir2022Jul_THz_CE_FL}, a federated learning (FL)-based THz channel estimation is proposed for wideband systems in the presence of beam-split. However, the algorithm in \cite{elbir2022Jul_THz_CE_FL} is based on the far-field user assumption, which may not be satisfied in THz-band scenario.

The aforementioned THz works~\cite{dovelos_THz_CE_channelEstThz2,elbir2022Jul_THz_CE_FL,elbir_THZ_CE_ArrayPerturbation_Elbir2022Aug,spatialWidebandWang2018May,elbir2022_thz_beamforming_Unified_Elbir2022Sep,thz_channelEst_beamsplitPatternDetection_L_Dai} as well as  the conventional wireless systems operating at sub-$6$ GHz and mm-Wave bands~\cite{limitedFeedback_Alkhateeb2015Jul,channelEstLargeArrays2,channelEstLargeArrays,channelEstimation1}  mostly incorporate far-field plane-wave model whereas the transmission range is shorter in THz-band such that the users are usually in the near-field region~\cite{ummimoTareqOverview}. Specifically, the plane wavefront is spherical in the near-field when the transmission range is shorter than the Fraunhofer distance~\cite{nf_primer_Bjornson2021Oct}. As a result, the channel acquisition algorithms should take into account near-field model (see, e.g., Fig.~\ref{fig_BS}), which depends on both direction and range information  for accurate signal processing~\cite{elbir2022Aug_THz_ISAC}. Among the works investigating the near-field signal  model,  \cite{nf_OMP_Dai_Wei2021Nov,nf_mmwave_CE_noBeamSplit_Cui2022Jan,nf_hybridNF_FF_Yu2022May,nf_NB_Zhang2022Dec,nf_NB2_Zhang2022Nov} consider the near-field scenario while the effect of beam-split is ignored and only mm-Wave scenario is investigated. In particular, \cite{nf_NB_Zhang2022Dec} considers the wideband mm-Wave channel estimation while the authors in \cite{nf_NB2_Zhang2022Nov} devise a machine learning (ML)-based approach for narrowband near-field channel estimation. In addition to near-field-only model, hybrid (near- and far-field) models are also present in the literature~\cite{nf_OMP_Dai_Wei2021Nov,nf_hybridNF_FF_Yu2022May}, wherein only narrowband transceiver architectures are considered. On the other hand,  several methods have been proposed to compensate the far-field beam-split for both THz channel estimation~\cite{elbir2022Jul_THz_CE_FL,elbir_THZ_CE_ArrayPerturbation_Elbir2022Aug,dovelos_THz_CE_channelEstThz2,thz_channelEst_beamsplitPatternDetection_L_Dai} and beamforming~\cite{elbir2021JointRadarComm,elbir2022_thz_beamforming_Unified_Elbir2022Sep,delayPhasePrecoding_THz_Dai2022Mar} applications.  Nevertheless, THz channel estimation in the presence of \textit{near-field beam-split (NB)} remains relatively unexamined.  

\begin{figure*}[h]
	\centering
	{\includegraphics[draft=false,width=\textwidth]{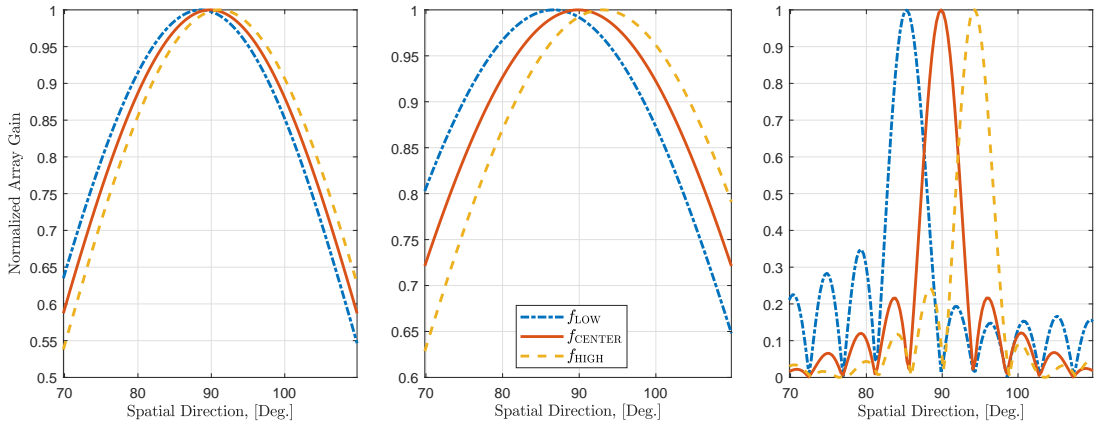} } 
	\caption{Normalized array gain with respect to spatial direction at low, center and high end subcarriers for (left) $3.5$ GHz, $B=0.1$ GHz; (middle) $28$ GHz, $B=2$ GHz; and (right) $300$ GHz, $30$ GHz, respectively.    	}
	\label{fig_ArrayGain_BS}
\end{figure*}

\subsection{Contributions}
In this work, we introduce both a model-based and a model-free techniques for near-field THz channel estimation in the presence of NB.  In the model-based approach, we propose an NB-aware (NBA) THz channel estimation technique based on OMP, henceforth called NBA-OMP. For the model-free approach, we devise a federated learning (FL) approach, which is more communication-efficient as compared to conventional centralized learning (CL)-based methods.

We design a novel NBA dictionary, whose columns are composed of near-field subcarrier-dependent (SD) steering vectors spanning the whole angular spectrum and the transmission ranges up to the Fraunhofer distance. Then, we introduce our proposed NBA-OMP  approach for wideband THz channel estimation.

The key idea of the proposed approach is that the degree of beam-split is proportionally known prior to the direction-of-arrival (DoA)/range estimation while it depends on the unknown user location. For example, consider $f_m$ and $f_c$ to be the frequencies for the $m$-th and center subcarriers, respectively. When $\theta$ is the physical user direction,  the spatial direction corresponding to the $m$-th-subcarrier is shifted by $\frac{f_c}{f_m}\theta$. Thus, we  employ the OMP algorithm, which accounts for this deviation thereby \textit{ipso facto} compensating the effect of NB.  We have conducted several numerical experiments to demonstrate the effectiveness of the proposed approach in comparison with the existing THz channel-estimation-based methods~\cite{nf_OMP_Dai_Wei2021Nov,thz_channelEst_beamsplitPatternDetection_L_Dai}.

\subsection{Outline and Notation}
\subsubsection{Outline} In the remainder of the paper, we first introduce the signal model for multi-user wideband THz ultra-massive MIMO system in Sec.~\ref{sec:probForm}. Then, we present the proposed NBA OMP technique for a model-based and a model-free channel estimation in Sec.~\ref{sec:NBAOMP} and Sec.~\ref{sec:FL}, respectively. The complexity and overhead analysis of the proposed approaches are discussed in Sec.~\ref{sec:Cpmplexity}. Sec.~\ref{sec:Sim} presents the numerical simulations and we finalize the paper in Sec.~\ref{sec:Conc} with concluding remarks.

\subsubsection{Notation}

Throughout the paper, we denote the vector and matrices via bold lowercase and uppercase letters, respectively. The transpose and conjugate transpose operations are denoted by $(\cdot)^\textsf{T}$ and $(\cdot)^{\textsf{H}}$, respectively. We denote the $n$-th column of a matrix $\mathbf{A}$ as $\mathbf{A}_n$ and  $\mathbf{A}^{\dagger}$ represents the Moore-Penrose pseudo-inverse. For a vector $\mathbf{a}$, the $n$-th element of $\mathbf{a}$ is represented by $[\mathbf{a}]_n$. $\lceil\cdot \rceil$ is the ceiling operator, $\nabla$ represents the gradient operation.  $\Sigma(a) = \frac{\sin \pi ab }{b\sin \pi a}$ denotes the Dirichlet sinc function, and $\mathbb{E}\{\cdot\}$ stands for the expectation operation.  $|| \cdot ||_2$ and $|| \cdot ||_\mathcal{F}$ denote the $l_2$ and Frobenius norms, respectively.

\section{System Model}
\label{sec:probForm}
Consider a wideband THz MIMO architecture with hybrid analog/digital beamforming over $M$ subcarriers. We assume that the base station (BS) has $N$ antennas and $N_\mathrm{RF}$ radio-frequency (RF) chains to serve $K$ single-antenna users. Let  $\mathbf{s}[m] = [s_1[m],\cdots,s_K[m]]^\textsf{T}$ denote the data symbols, where $m\in \mathcal{M} = \{1,\cdots, M\}$, which are processed via a $K\times K$ SD baseband beamformer $\mathbf{F}_\mathrm{BB}[m] = [\mathbf{f}_{\mathrm{BB}_1}[m],\cdots,\mathbf{f}_{\mathrm{BB}_K}[m]]$. In order to steer the generated beams toward users in downlink, an  $N\times N_\mathrm{RF}$ SI analog beamformer $\mathbf{F}_\mathrm{RF}$ ($N_\mathrm{RF}=K<{N}$) is employed. Since the analog beamformers are realized with phase-shifters, they have constant-modulus constraint, i.e., $|[\mathbf{F}_\mathrm{RF}]_{i,j}| = \frac{1}{\sqrt{N}}$. Then, the $N\times 1$ transmitted signal becomes
\begin{align}
\tilde{\mathbf{x}}[m] =  \mathbf{F}_\mathrm{RF}\mathbf{F}_\mathrm{BB}[m]\mathbf{s}[m].
\end{align}

The transmitted signal is passed through the wireless channel, and it is received by the $k$-th user at the $m$-th subcarrier as 
\begin{align}
\label{receivedSignal}
{y}_{k}[m] &= \mathbf{h}_k^\textsf{T}[m] \tilde{\mathbf{x}}[m] + {w}_k[m] \nonumber \\
&= \mathbf{h}_{k}^\textsf{T}[m]\sum_{i = 1}^{K}\mathbf{F}_\mathrm{RF}\mathbf{f}_{\mathrm{BB}_i}[m]{s}_i[m]  + {w}_k[m],
\end{align}
where ${w}_k[m]\in \mathbb{C}$ represents the complex additive white Gaussian noise with variance of $\sigma_n^2$, i.e., ${w}_k[m] \sim \mathcal{CN}({0},\sigma_n^2)$.

\subsection{THz Channel Model}
Due to limited reflected path components and negligible scattering, the THz channel is usually constructed as the superposition of a single LoS path with a few assisting NLoS paths~\cite{ummimoTareqOverview,elbir2021JointRadarComm,thz_beamSplit}.	In addition, multipath channel models are also widely used, especially for indoor applications~\cite{teraMIMO,ummimoTareqOverview}. Hence, we consider a general scenario, wherein the $N\times 1$ channel matrix for the $k$-th user at the $m$-th subcarrier is represented by the combination of $L$ paths as~\cite{ummimoTareqOverview}
\begin{align}
\label{channelModel}
\mathbf{h}_k[m]  =  
\sqrt{\frac{N}{L}}  \sum_{l =1}^{L}   \alpha_{k,m,l} \mathbf{a}(\phi_{k,l},r_{k,l})   e^{-j2\pi\tau_{k,l} f_m },
\end{align}
where  $\tau_{k,l}$ represents the time delay of the $l$-th path corresponding to the array origin.  $\alpha_{k,m,l}\in\mathbb{C}$  denotes the complex path gain and the expected value of its magnitude for the indoor THz multipath model is given by
\begin{align}
\mathbb{E}\{|\alpha_{k,m,l}|^2 \} = \left(\frac{c_0}{4\pi f_m r_{k,l} } \right)^2 e^{- k_\mathrm{abs}(f_m) r_{k,l}  },
\end{align}	
where  $f_m$ is the $m$-th subcarrier frequency, $c_0$ is  speed of light, $r_{k,l}$ represents the distance from the $k$-th user to the array origin and  $k_\mathrm{abs}(f_m)$ is the SD medium absorption coefficient~\cite{ummimoTareqOverview,ummimoHBThzSVModel,thz_clusterBased_Yuan2022Mar}. Furthermore, $f_m = f_c + \frac{B}{M}(m - 1 - \frac{M-1}{2}) $, where  $f_c$ and $B$ are carrier frequency and bandwidth, respectively.

\subsection{Near-field Array Model}
Due to high frequency operations as well as employing extremely small array aperture, THz-band transmission is likely to encounter near-field phenomenon for close-proximity users. Specifically,	the far-field model involves the reception of the transmitted signal at the users as plane-wave. However, the plane wavefront is spherical in the near-field if the transmission range is shorter
than the Fraunhofer distance $F = \frac{2 D^2}{\lambda}$, where $D$ is the array aperture and $\lambda = \frac{c_0}{f_c}$ is the wavelength~\cite{nf_primer_Bjornson2021Oct,elbir_THZ_CE_ArrayPerturbation_Elbir2022Aug}. The illustration of beam-split for both far-field and near-field models is given in Fig.~\ref{fig_BS}. For a uniform linear array (ULA), the array aperture is $D = (N-1)d$, where $d = \frac{\lambda}{2}$ is the element spacing. Thus, for THz-band transmission, this distance becomes small such that the near-field signal model should be employed since $r_{k,l} <F$. For instance, when $f_c = 300$ GHz and $N=256$, the Fraunhofer distance is $F = 32.76$ m.

We define the near-field steering vector $\mathbf{a}(\phi_{k,l},r_{k,l})\in\mathbb{C}^{N}$ corresponding to the physical DoA  $\phi_{k,l}$ and range $r_{k,l}$ as 
\begin{align}
\label{steeringVec1}
\mathbf{a}(\phi_{k,l},r_{k,l}) = \frac{1}{\sqrt{N}} [e^{- \mathrm{j}2\pi \frac{d}{\lambda}r_{k,l}^{(1)} },\cdots,e^{- \mathrm{j}2\pi \frac{d}{\lambda_m}r_{k,l}^{(N)} }]^\textsf{T},
\end{align}
where $\phi_{k,l} = \sin \tilde{\phi}_{k,l}$ with $\tilde{\phi}_{k,l}\in [-\frac{\pi}{2},\frac{\pi}{2}]$, and  $r_{k,l}^{(n)}$ is the distance between the $k$-th user and the $n$-th antenna as
\begin{align}
r_{k,l}^{(n)} = \left(r_{k,l}^2  + 2(n-1)^2 d^2 - 2 r_{k,l}(n-1) d \phi_{k,l}   \right)^{\frac{1}{2}}. \label{eq:rkln}
\end{align}
Following the Fresnel approximation ~\cite{nf_Fresnel_Cui2022Nov,nf_primer_Bjornson2021Oct}, \eqref{eq:rkln} becomes
\begin{align}
\label{r_approx}
r_{k,l}^{(n)} \approx r_{k,l}  - (n-1) d \phi_{k,l}  + (n-1)^2 d^2 \zeta_{k,l}  ,
\end{align}	 
where $\zeta_{k,l} = \frac{1- \phi_{k,l}^2}{2 r_{k,l}}$. Rewrite (\ref{steeringVec1}) as
\begin{align}
\label{steeringVectorPhy}
\mathbf{a}(\phi_{k,l},r_{k,l}) \approx e^{- \mathrm{j}2\pi \frac{f_c}{c_0}r_{k,l}} \tilde{\mathbf{a}}(\phi_{k,l},r_{k,l}),
\end{align} where the $n$-th element of $\tilde{\mathbf{a}}(\phi_{k,l},r_{k,l})\in \mathbb{C}^N$ is 
\begin{align}
\label{steeringVectorPhy2}
[\tilde{\mathbf{a}}(\phi_{k,l},r_{k,l})]_n = e^{\mathrm{j} 2\pi \frac{f_c}{c_0}\left( (n-1)d\phi_{k,l}  - (n-1)^2 d^2 \zeta_{k,l}\right) }.
\end{align}
The steering vector in (\ref{steeringVectorPhy}) corresponds to the physical location $(\phi_{k,l},r_{k,l})$, which deviates to the spatial location $(\bar{\phi}_{k,m,l},\bar{r}_{k,m,l})$ in the beamspace arising from the absent of SD analog beamformers. Then, the $n$-th entry of the deviated steering vector in (\ref{steeringVectorPhy2}) for the spatial location is 
\begin{align}
\label{steeringVectorSpa}
&[\tilde{\mathbf{a}}(\bar{\phi}_{k,m,l},\bar{r}_{k,m,l})]_n = e^{\mathrm{j} 2\pi \frac{f_m}{c_0}\left( (n-1)d\bar{\phi}_{k,m,l}  - (n-1)^2 d^2 \bar{\zeta}_{k,m,l}\right) }.
\end{align}

\subsection{Problem Formulation}
The aim of this work is to estimate the wideband THz channel $\mathbf{h}_k[m]$ in the presence of NB in downlink scenario. 

\textit{Assumption 1:} For the proposed model-based technique, we assume that the users are synchronized and they use the received pilot signals transmitted by the BS. We also assume that all users are in the near-field region of the BS, i.e., $r_{k,l} \leq F$.

\textit{Assumption 2:} For the proposed model-free technique, it is assumed that the datasets of all users are collected prior to the model training stage, and the labels of these datasets are determined by the model-based technique introduced in Sec.~\ref{sec:NBAOMP}.

In what follows, we first introduce the NB model, then present the proposed approaches for model-based and model-free channel estimation in the presence of NB.

\section{NB Model}
Compared to mm-Wave frequencies, in THz-band, the bandwidth is so wide that a single-wavelength assumption for  beamforming cannot hold and it leads to the split of physical DoA/ranges $\{\phi_{k,l},r_{k,l}\}$ in the spatial domain. Hence, we define the relationship between the physical and spatial DoA/ranges in the following theorem:

\begin{theorem}
	Denote $\mathbf{u}\in \mathbb{C}^N $ and $\mathbf{v}_m \in \mathbb{C}^N$ as the arbitrary near-field steering vectors corresponding to the physical (i.e., $\{\phi_{k,l},r_{k,l}\}$) and spatial (i.e., $\{\bar{\phi}_{k,m,l},	\bar{r}_{k,m,l}\}$) locations given in (\ref{steeringVectorPhy2}) and (\ref{steeringVectorSpa}), respectively. Then, in spatial domain at subcarrier frequency $f_m$, the array gain achieved by $\mathbf{u}^\textsf{H}\mathbf{v}_m$ is maximized and the generated beam is focused at the location $\{\bar{\phi}_{k,m,l},	\bar{r}_{k,m,l}\}$ such that  
	\begin{align}
	\label{physical_spatial_directions}
	\bar{\phi}_{k,m,l} =    \eta_m \phi_{k,l}, \hspace{5pt}
	\bar{r}_{k,m,l} =    \frac{1 - \eta_m^2 \phi_{k,l}^2}{\eta_m(1 -\phi_{k,l}^2)}r_{k,l},
	\end{align}
	where  	 $\eta_m = \frac{f_c}{f_m}$ represents the proportional deviation of DoA/ranges.
\end{theorem}

\begin{proof}
	Define the array gain achieved by $\mathbf{v}_m$ on an arbitrary user location $\{\phi_{k,l},r_{k,l}\}$ with steering vector $\mathbf{u}$ as
	\begin{align}
	&G(\phi_{k,l},r_{k,l},{m}) = \frac{|\mathbf{u}^\textsf{H} \mathbf{v}_m|^2}{N^2} \nonumber\\
	& = \frac{1}{N^2} \left|\sum_{n = 0}^{N-1} e^{\mathrm{j}\frac{2\pi}{c_0}  \left[ nd\left(f_m\bar{\phi}_{k,m,l} - f_c \phi_{k,l} \right) - n^2d^2(f_m\bar{\zeta}_{k,m,l} - f_c \zeta_{k,l} )  \right]    }   \right|^2 \nonumber\\
	&=  \frac{1}{N^2} \left|\sum_{n = 0}^{N-1} e^{\mathrm{j}\frac{2\pi n}{c_0}  (f_m\kappa_m -  f_c\kappa )    }   \right|^2 \nonumber\\
	&= \frac{1}{N^2} \left| \frac{1 - e^{-\mathrm{j}2\pi N (f_m\kappa_m -  f_c\kappa)   }}{1 - e^{-\mathrm{j}2\pi (f_m\kappa_m -  f_c\kappa)  }}   \right|^2  \nonumber\\
	&= \frac{1}{N^2}\left| \frac{\sin (\pi N(f_m\kappa_m -  f_c\kappa) )}{\sin (\pi (f_m\kappa_m -  f_c\kappa))}    \right|^2 \nonumber \\
	& = |\Sigma( f_m\kappa_m -  f_c\kappa )|^2. \label{arrayGain}
	\end{align}
	
	Furthermore, we define
	\begin{align}
	\bar{\zeta}_{k,m,l} &=  \eta_m \zeta_{k,l}  = \frac{1 - \bar{\phi}_{k,m,l}^2}{2 \bar{r}_{k,m,l}}, \\
	\kappa_m &= d(\bar{\phi}_{k,m,l} - n d \bar{\zeta}_{k,m,l}), \\
	\kappa &= d({\phi}_{k,l} - n d \zeta_{k,l})
	\end{align}
	The array gain in (\ref{arrayGain}) implies that most of the power is focused only on a small portion of the beamspace due to the power-focusing capability of $\Sigma(a)$, which substantially reduces across the subcarriers as $|f_m - f_c|$ increases. Furthermore, $|\Sigma( a)|^2$ gives peak when $a = 0$, i.e., $f_m\bar{\phi}_{k,m,l} =  f_c\phi_{k,l}$ and $f_m\bar{\zeta}_{k,m,l} = f_c\zeta_{k,l}$. Therefore, we have 
	\begin{align}
	\bar{\phi}_{k,m,l} = \eta_m \phi_{k,l} .
	\end{align}
	Then, by using $	\bar{\zeta}_{k,m,l} =  \eta_m \zeta_{k,l}$, we get
	\begin{align}
	\label{physical_spatial_ranges}
	\bar{r}_{k,m,l} =   \frac{1 - \eta_m^2 \phi_{k,l}^2}{\eta_m(1 -\phi_{k,l}^2)}r_{k,l}.
	\end{align}
\end{proof}

Finally, by combining (\ref{r_approx}), (\ref{physical_spatial_directions}) and (\ref{physical_spatial_ranges}),  we define the NB in terms of DoAs and the ranges of the users  as
\begin{align}
\label{beamSplit2}
\Delta(\phi_{k,l},m) &= \bar{\phi}_{k,m,l} - \phi_{k,l} = (\eta_m -1)\phi_{k,l}, \\
\Delta(r_{k,l},m) &= \bar{r}_{k,m,l} - r_{k,l} = \left(\frac{1 - \eta_m^2 \phi_{k,l}^2}{\eta_m(1 -\phi_{k,l}^2)} - 1 \right) r_{k,l}.
\end{align}

\section{Model-Based Solution: NBA-OMP}
\label{sec:NBAOMP}

We first introduce how the NBA dictionary is designed, then present the implementation steps of the proposed NBA-OMP technique.

\subsection{NBA Dictionary Design}
The key idea of the proposed NBA dictionary design is to utilize the prior knowledge of $\eta_m$ to obtain beam-split-corrected steering vectors. In other words, for an arbitrary physical DoA and range, we can readily find the spatial DoA and ranges as $\bar{\phi} = \eta_m\phi$ and $\bar{r} = \frac{1 - \eta_m^2 \phi^2}{\eta_m(1 - \phi^2)} r$. Using this observation, we design the NBA dictionary $\mathcal{C}_m$ composed of steering vectors $\mathbf{c}(\phi_m,r_m)\in\mathbb{C}^N$ as
\begin{align}
\mathcal{C}_m = \{\mathbf{c}(\phi_m,r_m) |  \phi_m\in [-\eta_m,-\eta_m], r_m \in \mathbb{R}^+  \},
\end{align}
where the $n$-th element of $\mathbf{c}(\phi_m,r_m)$ is 
\begin{align}
\label{steeringVec2c}
[\mathbf{c}(\phi_m,r_m)]_n =  e^{\mathrm{j} 2\pi \frac{f_m}{c_0}\left(  (n-1)d\phi_m  - (n-1)^2 d^2 \zeta_m\right) },
\end{align}
where $\zeta_m= \frac{1 - \phi_m^2}{2r_m}  $.

Using the NBA dictionary $\mathcal{C}_m$, instead of SI steering vectors $\mathbf{a}(\phi,r)$, the SD virtual steering vectors $\mathbf{c}(\phi_m,r_m)$ can be constructed for the OMP algorithm. Once the beamspace spectra is computed via OMP, the sparse channel support corresponding to the SD spatial DoA and ranges are obtained. Then, one can readily find the physical DoAs and ranges as $\phi = \phi_m/\eta_m$ and ${r} =    \frac{\eta_m(1 - \phi^2)}{1 - \eta_m^2 \phi^2}r_m$,  $\forall m\in \mathcal{M}$.


The proposed NBA dictionary $\mathcal{C}_m$ also holds spatial orthogonality as
\begin{align}
\lim_{N \rightarrow +\infty} |\mathbf{c}^\textsf{H}(\phi_{m,i},r_{m,i}) \mathbf{c}(\phi_{m,j},r_{m,j})  | = 0, \forall  i\neq j.
\end{align}
In the next section, we present the channel estimation procedure with the proposed NBA dictionary.


\subsection{Near-field Channel Estimation}
In downlink, the channel estimation stage is performed simultaneously  during channel training by all the users, i.e., $k\in \mathcal{K} = \{1,\cdots, K\}$. Since the BS employs hybrid beamforming architecture, it activates only a single RF chain in each channel use to transmit the pilot signals during channel acquisition~\cite{heath2016overview}. Denote $\tilde{\mathbf{S}}[m] = \mathrm{diag}\{\tilde{s}_1[m],\cdots, \tilde{s}_P[m]\}$ as the $P\times P$ orthogonal pilot signal matrix, then the BS employs $P$ beamformer vectors as $\tilde{\mathbf{F}}= [\tilde{\mathbf{f}}_1,\cdots, \tilde{\mathbf{f}}_P]\in \mathbb{C}^{N\times P}$ ($|\tilde{\mathbf{f}}_p| = 1/\sqrt{N}$) to send $P$ orthogonal pilots, which are collected by the $k$-th user as
\begin{align}
\mathbf{y}_k[m] = \tilde{\mathbf{S}}[m]\bar{\mathbf{F}}[m] \mathbf{h}_k[m] + \mathbf{w}_k[m],
\end{align}
where $\bar{\mathbf{F}}[m] = \tilde{\mathbf{F}}^\textsf{T}[m]\in \mathbb{C}^{P\times N} $. Assume that $\bar{\mathbf{F}}[m] = \mathbf{F}\in \mathbb{C}^{P\times N}$ and $\tilde{\mathbf{S}}[m]= \mathbf{I}_P,\forall m\in \mathcal{M}$~\cite{limitedFeedback_Alkhateeb2015Jul,dovelos_THz_CE_channelEstThz2,thz_channelEst_beamsplitPatternDetection_L_Dai}, we get
\begin{align}
\label{y_vector}
\mathbf{y}_k[m] = \mathbf{F}\mathbf{h}_k[m] + \mathbf{w}_k[m].
\end{align}
Note that the solution via traditional techniques, e.g.,  the least squares (LS) and MMSE estimator can be readily given respectively as
\begin{align}
\mathbf{h}_k^{\mathrm{LS}}[m] &= ({\mathbf{F}}^\textsf{H}{\mathbf{F}})^{-1}{\mathbf{F}}^\textsf{H}{\mathbf{y}}_k[m]  ,  \\
\mathbf{h}_k^{\mathrm{MMSE}}[m] &= \big(\mathbf{R}_k^{-1}[m] + {\mathbf{F}}^\textsf{H} {\mathbf{R}}_k^{-1}[m]{\mathbf{F}}  \big)^{-1}{\mathbf{F}}^\textsf{H}{\mathbf{y}}_k[m], 	\label{ls_MMSE}
\end{align}
where $\mathbf{R}_k[m] = \mathbb{E}\{\mathbf{h}_k [m]\mathbf{h}_k^{\textsf{H}}[m]\}$ is the channel covariance matrix. Nevertheless, these methods require either prior information (e.g., MMSE) or have poor estimation performance (e.g., LS)~\cite{elbir2020_FL_CE,channelEstimation1}. Furthermore, these methods require at least $P \geq N$ pilot signals, which can be heavy for channel training while our proposed approach exhibits a much lower channel training overhead (e.g, see Sec.~\ref{sec:Cpmplexity}).

By exploiting the THz channel sparsity, the THz channel can be represented in a sparse domain via support vector $\mathbf{x}_k[m]\in\mathbb{C}^{Q}$, which is an $L$-sparse vector, whose non-zero elements corresponds to the set
\begin{align}
\{ x_{k,l}[m]| x_{k,l}[m]\triangleq\sqrt{\frac{N}{L}} \alpha_{k,m,l}e^{-j2\pi \tau_{k,l}f_m} \}.
\end{align} 
Then, the received signal in (\ref{y_vector}) is rewritten in sparse domain as
\begin{align}
\mathbf{y}_k[m] = \mathbf{F} \mathbf{C}_m \mathbf{x}_k[m] + \mathbf{w}_k[m],
\end{align}
where $\mathbf{F}$ is an $P \times N$ fixed matrix corresponding to the hybrid beamformer weights during data transmission, and $\mathbf{C}_m$ is the ${N\times Q}$ NBA dictionary matrix covering the spatial domain with $\phi_{m,q} \in [-\eta_m,\eta_m]$ ($\phi \in [-1,1]$) and $r_{m,q}\in [0, F]$ for $q=1,\cdots, Q$ as
\begin{align}
\mathbf{C}_m = [\mathbf{c}(\phi_{m,1},r_{m,1}),\cdots,\mathbf{c}(\phi_{m,Q},r_{m,Q}) ].
\end{align}

By utilizing the NBA dictionary matrix $\mathbf{C}_m$, we employ the OMP algorithm to effectively recover the sparse channel support. The proposed NBA-OMP technique is presented in Algorithm~\ref{alg:BSACE}, which accepts $\mathbf{y}_k[m]$, $\mathbf{C}_m$ and $\eta_m$ as inputs and it yields the output as the estimated channel $\hat{\mathbf{h}}_k[m]$, and the near-field beam-split in terms of DoA and ranges, i.e.,  $\hat{\Delta}(\phi_{k,l},m)$ and $\hat{\Delta}(r_{k,l},m)$, respectively. In steps $4-6$ of Algorithm~\ref{alg:BSACE}, the orthogonality between the residual observation vector $\mathbf{r}_l[m]$ and the columns of the dictionary matrix is checked as
\begin{align}
q^\star = \argmax_q \sum_{m=1}^{M}|\mathbf{c}^\textsf{H}(\phi_{m,q},r_{m,q})\mathbf{F}^\textsf{H}\mathbf{r}_{l-1}[m] |,
\end{align}
where $q^\star$ denotes the support index for the $l$th path iteration. Then, the DoA and range estimates can be found as $\hat{\phi}_{k,l} =\frac{\phi_{m,q^\star}}{\eta_m}$ and  $\hat{r}_{k,l} =    \frac{\eta_m(1 - \phi_{m,q^\star}^2)}{1 - \eta_m^2 \phi_{m,q^\star} } r_{m,q^\star}$, respectively. In a similar way, the corresponding beam-split terms are $\hat{\Delta}(\phi_{k,l},m) = (\eta_m-1) \hat{\phi}_{k,l}$ and $\hat{\Delta}(r_{k,l},m) \hspace{-3pt}=\hspace{-2pt}(\eta_m\hspace{-2pt} -\hspace{-2pt}1) \frac{1 - \eta_m^2 \hat{\phi}_{k,l}^2}{\eta_m(1 -\hat{\phi}_{k,l}^2)}\hat{r}_{k,l} $. Once the DoA and ranges are obtained, the THz channel is constructed from the estimated support as $\hat{\mathbf{h}}_k[m] = \boldsymbol{\Xi}_k\hat{\mathbf{u}}_k[m]$ where $\boldsymbol{\Xi}_k = [\mathbf{a}(\hat{\phi}_{k,1},\hat{r}_{k,1}), \cdots, \mathbf{a}(\hat{\phi}_{k,L}, \hat{r}_{k,L})]$ and $\hat{\mathbf{u}}_k[m] = \boldsymbol{\Psi}_m^\dagger(\mathcal{I}_{L}) \mathbf{y}_k[m]$, wherein $\boldsymbol{\Psi}_m(\mathcal{I}_L) = \mathbf{F}\mathbf{C}_m(\mathcal{I}_{L})$ and $\mathcal{I}_{L} $ includes the estimated support indices.


\begin{algorithm}[t]
	\begin{algorithmic}[1] 
		\caption{ \bf NBA-OMP}
		\Statex {\textbf{Input:}  Dictionary $\mathbf{C}_m$, observation $\mathbf{y}_k[m]$ and $\eta_m$ \label{alg:BSACE}}
		\Statex \textbf{Output:}   $\hat{\mathbf{h}}_k[m]$, $\hat{\Delta}(\phi_{k,l},m)$ and $\hat{\Delta}(r_{k,l},m)$.
		\State \textbf{for} $k \in\mathcal{K}$
		\State \indent $l=1$, $\mathcal{I}_{l-1} = \emptyset$,
		$\mathbf{r}_{l-1}[m] = \mathbf{y}_k[m], \forall m\in \mathcal{M}$.
		\State \indent\textbf{while} $l\leq L$ \textbf{do}
		
		\State \indent\indent $q^\star = \argmax_q \sum_{m=1}^{M}|\mathbf{c}^\textsf{H}(\phi_{m,q},r_{m,q})\mathbf{F}^\textsf{H}\mathbf{r}_{l-1}[m] |$.
		\State \indent\indent $\mathcal{I}_{l} = \mathcal{I}_{l-1} \bigcup \{q^\star\}$.
		\State \indent \indent  $\hat{\phi}_{k,l} =\frac{\phi_{m,q^\star}}{\eta_m}$.
		\State \indent \indent $\hat{r}_{k,l} =    \frac{\eta_m(1 - \phi_{m,q^\star}^2)}{1 - \eta_m^2 \phi_{m,q^\star} } r_{m,q^\star}$.
		\State \indent \indent $\hat{\Delta}(\phi_{k,l},m) = (\eta_m-1) \hat{\phi}_{k,l}$, $\forall m\in \mathcal{M}$.
		\State \indent \indent $\hat{\Delta}(r_{k,l},m) \hspace{-3pt}=\hspace{-2pt}(\eta_m\hspace{-2pt} -\hspace{-2pt}1) \frac{1 - \eta_m^2 \hat{\phi}_{k,l}^2}{\eta_m(1 -\hat{\phi}_{k,l}^2)}\hat{r}_{k,l} $, $\forall m\in \mathcal{M}$.
		\State \indent \indent $\boldsymbol{\Psi}_m(\mathcal{I}_l) = \mathbf{F}\mathbf{C}_m(\mathcal{I}_1)$.
		\State \indent \indent $\mathbf{r}_{l}[m] = \left( \mathbf{I}_P -  \boldsymbol{\Psi}_m(\mathcal{I}_l) \boldsymbol{\Psi}_m^\dagger(\mathcal{I}_l) \right) \mathbf{y}_k[m]$.
		\State \indent\indent $l= l + 1$.
		\State \indent\textbf{end while}
		\State \indent  $\boldsymbol{\Xi}_k = [\mathbf{a}(\hat{\phi}_{k,1},\hat{r}_{k,1}), \cdots, \mathbf{a}(\hat{\phi}_{k,L}, \hat{r}_{k,L})]$.
		\State \indent \textbf{for} $m\in \mathcal{M}$
		\State \indent\indent $\hat{\mathbf{u}}_k[m] = \boldsymbol{\Psi}_m^\dagger(\mathcal{I}_{l-1}) \mathbf{y}_k[m]$. \\
		\indent \indent \indent $\hat{\mathbf{h}}_k[m] = \boldsymbol{\Xi}_k\hat{\mathbf{u}}_k[m]$.
		\State \indent\textbf{end for}
		\State \textbf{end for}
	\end{algorithmic} 
\end{algorithm}

\section{Model-Free Solution: NBA-OMP-FL}
\label{sec:FL}
Due to operating at small wavelengths and severe path loss, ultra-massive arrays are envisioned to be employed at THz communications. This results in very large number of array data to be processed. For instance, when employing learning-based techniques, the size of the datasets is proportional to the number of antennas. Thus, huge communication overhead occurs due to the transmission of these datasets in CL-based techniques~\cite{elbir_FL_PHY_Elbir2021Nov,elbir2022Jul_THz_CE_FL,fl_By_Google}. To alleviate this overhead, we present an FL-based near-field THz channel estimation scheme in this section.

Denote $\boldsymbol{\theta}\in\mathbb{R}^Z$  and $\mathcal{D}_k$ as the vector of learnable parameters and the local dataset for the $k$-th user, respectively. Then, we have the  relationship $f(\boldsymbol{\theta})$  between the input of the learning model represented by $\boldsymbol{\theta}$ and its prediction as $	\mathcal{Y}_k^{(i)} = f(\boldsymbol{\theta}) \mathcal{X}_k^{(i)}$. Here,  the sample index for the local dataset is $i = 1,\dots, \textsf{D}_k$, where $\textsf{D}_k = |\mathcal{D}_k|$ denotes the number of samples in the $k$-th local dataset. Furthermore, we have $\mathcal{X}_k^{(i)}$ and $\mathcal{Y}_k^{(i)}$ representing the input and output, respectively,  such that the input-output tuple of the local dataset is $\mathcal{D}_k^{(i)} = (\mathcal{X}_k^{(i)},\mathcal{Y}_k^{(i)})$. 

The input data is the received pilots, i.e., $\mathbf{y}_k[m]$. For the output (label) data, we consider the channel estimates obtained via the proposed NBA-OMP approach in Algorithm~\ref{alg:BSACE}. Then, by using $\hat{\mathbf{h}}_k[m]$, the output data is constructed as
\begin{align}
\mathcal{Y}_k = \left[\operatorname{Re}\{\hat{\mathbf{h}}_k[m]\}^\textsf{T},\operatorname{Im}\{\hat{\mathbf{h}}_k[m]\}^\textsf{T}\right]^\textsf{T}\in \mathbb{R}^{2N}.
\end{align}

When designing the input data, the real, imaginary and angle information of $\mathbf{y}_k[m]$ are incorporated in order to improve the feature extraction performance of the learning model. Therefore, each input sample has a ``three-channel'' structure with the size of  ${N_\mathrm{RF}\times 3}$~\cite{elbir_FL_PHY_Elbir2021Nov}. Thus, we have  $[\mathcal{X}_k]_1 = \operatorname{Re}\{\mathbf{y}_k[m]\}$, $[\mathcal{X}_k]_2 = \operatorname{Im}\{\mathbf{y}_k[m]\}$ and $[\mathcal{X}_k]_3 = \angle\{\mathbf{y}_k[m]\}$. 

In order for all the users participate in the training in a distributed manner, each-user computes their model parameters and aims to minimize an average MSE cost. Thus, the FL-based model training problem can be cast as
\begin{align}
\label{flTraining}
\minimize_{\boldsymbol{\theta}} \hspace{5pt} &\frac{1}{K}\sum_{k=1}^K \mathcal{L}_k(\boldsymbol{\theta}) \nonumber \\
\subjectto \hspace{3pt} &f(\mathcal{X}_k^{(i)}| \boldsymbol{\theta}) = \mathcal{Y}_k^{(i)},
\end{align}
where $\mathcal{L}_k(\boldsymbol{\theta})$ is the loss function for the learning model computed at the $k$-th user, and it is defined as  
\begin{align}
\mathcal{L}_k (\boldsymbol{\theta}) = \frac{1}{\textsf{D}_k} \sum_{i=1}^{\textsf{D}_k} ||f(\mathcal{X}_k^{(i)}| \boldsymbol{\theta}) - \mathcal{Y}_k^{(i)} ||_{\mathcal{F}}^2,
\end{align}
where  $f(\mathcal{X}_k^{(i)}| \boldsymbol{\theta})$ corresponds to the output of the learning model once it is fed with $\mathcal{X}_k^{(i)}$. In order to efficiently solve the optimization problem presented in (\ref{flTraining}), iterative techniques can be conducted such that the learning model parameters are computed at each iteration based on local datasets. Specifically, first the learning model parameters are computed at the users, and then they are transmitted to the server, at which they are aggregated. Therefore, the model parameter update rule for the $t$-th iteration is described by
\begin{align}
\boldsymbol{\theta}_{t+1} = \boldsymbol{\theta}_t - \varepsilon \frac{1}{K} \sum_{k=1}^{K}\boldsymbol{\beta}_k(\boldsymbol{\theta}_t),
\end{align} 
for $t = 1,\dots, T$. Here,  $\boldsymbol{\beta}_k(\boldsymbol{\theta}_t) = \nabla \mathcal{L}_k(\boldsymbol{\theta}_t)\in\mathbb{R}^Z$  is the gradient vector and $\varepsilon$ is the learning rate.

During model training, each user simply transmits their model update vector $\boldsymbol{\beta}_k(\boldsymbol{\theta}_t)$ to the BS for $k \in \mathcal{K}$. Assume that the model  $\boldsymbol{\beta}_k(\boldsymbol{\theta}_t)$ is divided into $U = \frac{Z}{M}$ equal-length blocks for each subcarrier as
\begin{align}
\boldsymbol{\beta}_k(\boldsymbol{\theta}_t) = [\boldsymbol{x}_K^{(t)^\textsf{T}}[1],\dots,\boldsymbol{x}_K^{(t)^\textsf{T}}[M]]^\textsf{T},
\end{align}
where $\boldsymbol{x}_k^{(t)}[m] \in \mathbb{C}^{U }$ denotes data symbols (model updates) of the $k$-th user, that are transmitted to the server to be aggregated. The BS, then, receives the following $N\times U$ uplink signal as
\begin{align}
\boldsymbol{Y
}^{(t)}[m] = \sum_{k=1}^K \boldsymbol{h}_k^{(t)}[m]\boldsymbol{x}_k^{(t)^\textsf{T}}[m] + \boldsymbol{N}[m],
\end{align}
where $\boldsymbol{h}_k^{(t)}[m]  \in \mathbb{C}^{N}$  represents the uplink channel during model transmission and $\boldsymbol{N}[m]\in \mathbb{C}^{N\times U}$ denotes the noise term added onto the transmitted data.


\section{Complexity and  Overhead}
\label{sec:Cpmplexity}
\subsection{Computational Complexity}
The computational complexity of NBA-OMP is the same as traditional OMP techniques~\cite{heath2016overview}, and it is mainly due to the matrix multiplications in step $4$ $(O(QMNP(NP + P)))$, step $10$ $(O(PNL\bar{L}))$, step $11$ $(O(P^2L\bar{L} + P^2   ))$ and step $16$  $(O(L(P + N)))$, where $\bar{L}=(L+1)/2$. Hence, the overall complexity is $O( QMNP(NP + P) + (PL\bar{L} + 1)(N + P) )$.

\subsection{Channel Training Overhead} The channel training overhead of the NBA-OMP  requires only $P$ ($8$ times lower, see Sec.~\ref{sec:Sim}) channel usage for pilot signaling while the traditional approaches, e.g., LS and MMSE estimation, need at least $N$ channel usage.

\begin{table}[t]
	\caption{Learning Model Parameters}
	\label{tableParameters}
	\centering
	\begin{tabular}{|c|c|c|}
		\hline
		\hline
		\bf Layer&\bf Type&\bf Parameters      \\
		\hline
		1 &	Input & Size: $N_\mathrm{RF}\times3$ 	  \\
		\hline
		2 &	Convolutional & Size: $3\times3$ @128 	  \\
		\hline
		3 &	Normalization & -	  \\
		\hline
		4 &	Convolutional & Size: $3\times3$ @128 	  \\
		\hline
		5 &	Normalization & -	  \\
		\hline
		6 &	Convolutional & Size: $3\times3$ @128 	  \\
		\hline
		7 &	Normalization & -	  \\
		\hline
		8 &	Fully Connected & Size: $1024\times 1$	  \\
		\hline
		9 &	Dropout & Ratio: $0.5$	  \\
		\hline
		10 &Fully Connected & Size: $1024\times 1$	\\
		\hline
		11 &	Dropout & Ratio: $0.5$	  \\
		\hline
		12 &	Output (Regression) & Size: $2N\times 1$	\\
		\hline
		\hline 
	\end{tabular}
\end{table}

\subsection{Model Training Overhead}
The communication overhead of a learning-based approach can be measured in terms of amount of data symbols transmitted during/for model training~\cite{elbir2020_FL_CE,elbir2022Jul_THz_CE_FL,fl_By_Google,FL_gunduz_fading,FL_Gunduz}.  Thus, the communication overhead (i.e., $\mathcal{T}_\mathrm{CL}$) for the CL involves the amount of dataset transmitted from the users to the server whereas it can be computed for FL (i.e., $\mathcal{T}_\mathrm{FL}$) as the amount of model parameters exchanged between the users and the server during training. Then, we have 
\begin{align}
\label{tcl}
\mathcal{T}_\mathrm{CL} = \sum_{k = 1}^{K}\textsf{D}_k(3N_\mathrm{RF} + \xi 2N),
\end{align}
where $\textsf{D}_k$ is the number of samples of the $k$-th dataset and $3N_\mathrm{RF} + 2N$ is the number of symbols in each input-output (i.e., $\mathcal{X}_k^{(i)}-\mathcal{Y}_k^{(i)}$) tuple. In (\ref{tcl}), $\xi$ denotes whether the labels are also transmitted from the users to the server. Thus, if the labels are computed at the users, they are transmitted to the server with additional overhead, and we have $\xi = 1$. On the other hand,  $\xi = 0$ (labels are not computed at the user) if the received inputs are transmitted to the server, at which the labeling (channel estimation) is handled, thereby less communication overhead is achieved. Next, we defined the communication overhead for FL as
\begin{align}
\mathcal{T}_\mathrm{FL} = 2 Z TK,
\end{align}
which involves a two-way (user $\rightleftarrows$ server)  transmission of the learnable  parameters from $K$ users to the server for $T$ consecutive iterations.

\begin{figure*}[t!]
	\centering		{\includegraphics[draft=false,width=1\textwidth]{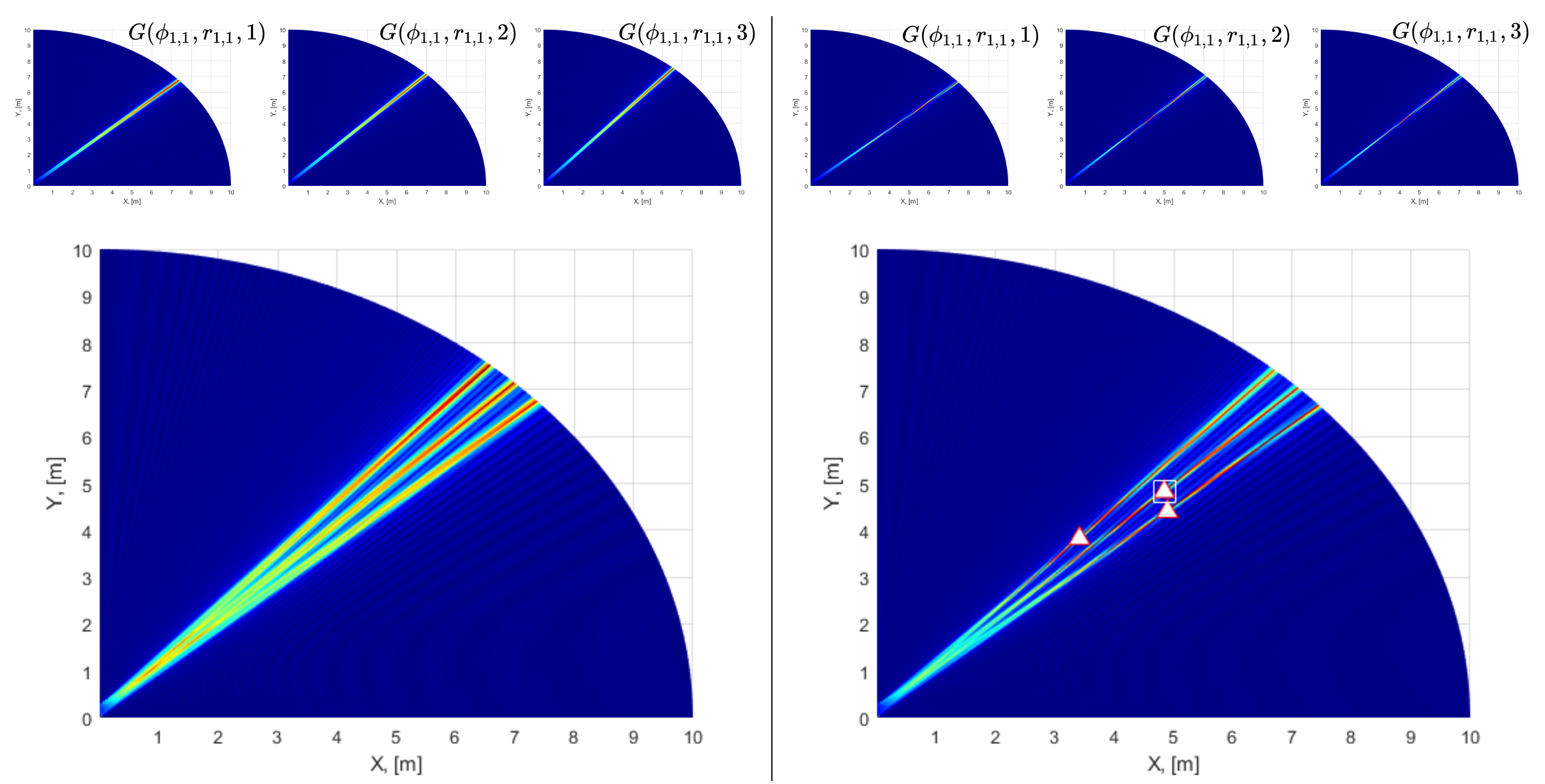}} 
	\caption{Array gains $G(\phi_{1,1},r_{1,1},{m})$ in Cartesian coordinates for a single user ($K=1$, $L=1$) located in the far-field $(45^\circ,6000\text{m})$ (left) and near-field $(45^\circ,6\text{m})$ (right), respectively. Here, $M=3$, $f_c=300$ GHz, and $B=30$ GHz. The top panel shows the gain for different subcarriers which are summed up to produce a composite array gain at the bottom for both far- and near-field cases clearly showing the beam-split. The square represents the user location while the triangles correspond to the spatial locations (where the maximum array gain is achieved) at different subcarriers. Whereas the far-field beam-split is only angular, the near-field split is across both range and angular domains.
		\vspace{-8pt}}
	\label{fig_BS}
\end{figure*}

\section{Numerical Simulations}
\label{sec:Sim}
We evaluate the performance of our proposed channel estimation approaches, in comparison with the state-of-the-art channel estimation techniques, e.g., far-field OMP (FF-OMP)~\cite{beamSquintRodriguezFernandez2018Dec}, near-field OMP (NF-OMP)~\cite{nf_OMP_Dai_Wei2021Nov}, beam-split pattern detection (BSPD)~\cite{thz_channelEst_beamsplitPatternDetection_L_Dai} as well as LS and MMSE. Note that, during the simulations, the MMSE estimator is computed subcarrier-wise such that it has a beam-split-free benchmark performance. 

Throughout the simulations, unless stated otherwise, the signal model is constructed with $f_c=300$ GHz, $B=30$ GHz, $M=128$, $K=N_\mathrm{RF} = 8$, $L=3$, $P=8$ and $N=256$, for which the Fraunhofer distance is $F = 32.76$ m. The NBA dictionary matrix is constructed with $Q=10N$, and the user directions and ranges are selected as $\tilde{\phi}_{k,l}\in \mathrm{unif}[-\frac{\pi}{2},\frac{\pi}{2}]$, $r_{k,l} \in \mathrm{unif}[5,30]$ m, respectively. 

A convolutional neural network (CNN) is designed with parameters enlisted in Table~\ref{tableParameters} through a hyperparameter optimization~\cite{FL_Gunduz,fl_By_Google}. The learning model is, then, trained for $T=100$ iterations, which involves the exchange of model parameter updates between the users and the server. The learning rate is selected as $\varepsilon= 0.001$. For  dataset generation, we first generated $V=1000$ independent channel realizations. Then, synthetic noise with AWGN is added onto the generated channel data in order to resemble the characteristics related to the imperfections/distortions in the wireless channels. This synthetic noise level is obtained over pre-determined noise level for three signal-to-noise ratio (SNR) levels, i.e., $\mathrm{SNR}_\mathrm{TRAIN} = \{15,20,25\}$ dB such that $G = 1000$ realizations of each channel data are obtained in order to improve robustness~\cite{elbir2020_FL_CE}. During channel generation, we considered the local datasets with non-identical distributions such that  the $k$-th user's dataset is selected with the DoA information as  $\vartheta_{k,l} \in [-\frac{\pi}{2} + \frac{\pi}{K} (k-1), -\frac{\pi}{2} + \frac{\pi}{K}k)$.

\begin{figure}[t]
	\centering		{\includegraphics[draft=false,width=.5\columnwidth]{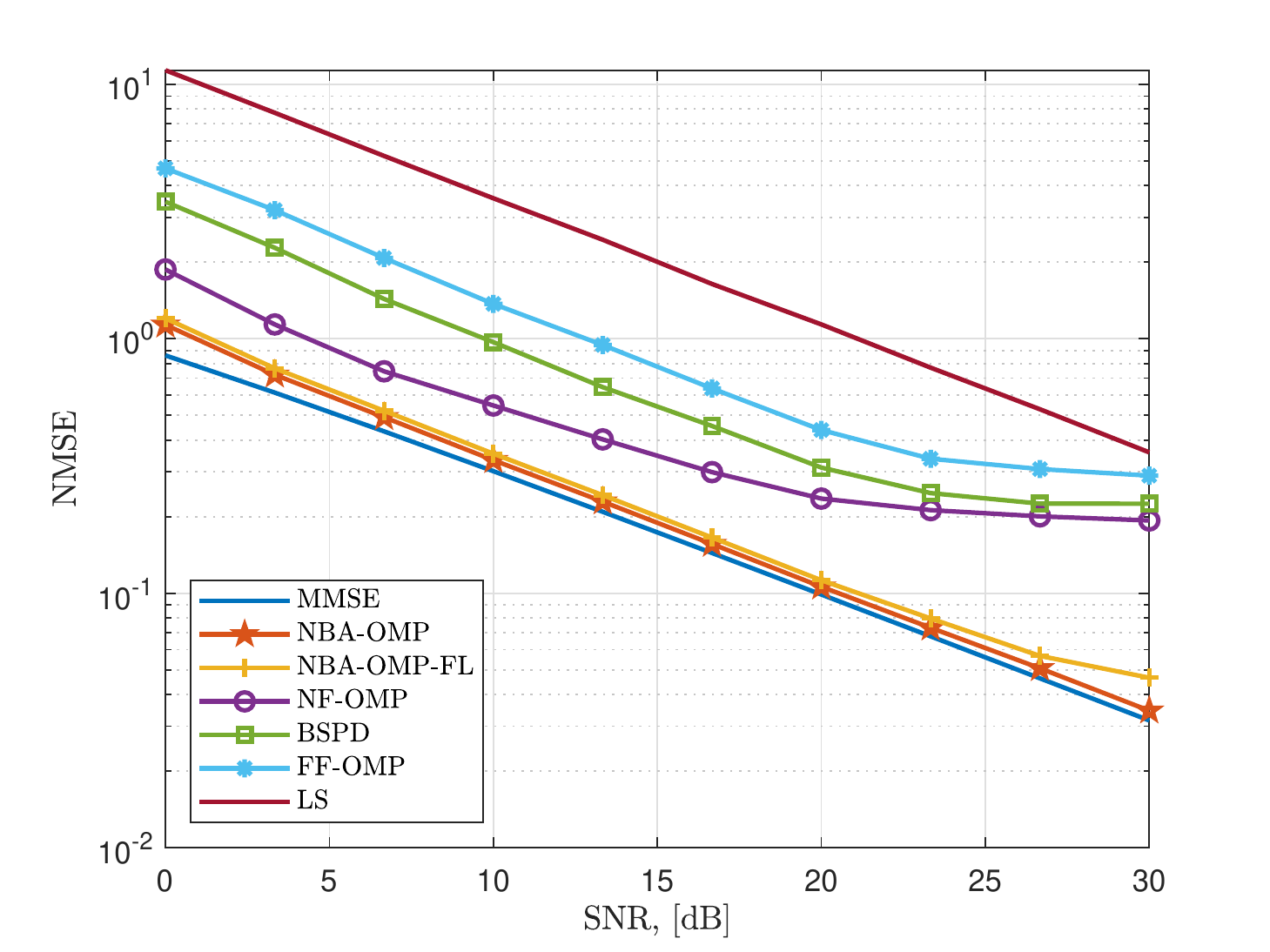}}  
	\caption{Near-field THz wideband channel estimation NMSE versus SNR. $N=256$, $f_c=300$ GHz, $M=128$ and $B=30$ GHz. 
	}
	\label{fig_NMSE_SNR}
\end{figure}

Fig.~\ref{fig_BS} shows the array gain in Cartesian coordinates for both far- and near-field scenarios, respectively.  The transmitter is located at $(0,0)$ m while the the user is located at $45^\circ$ with the distance of $6000$ m (far-field) and $8$ m (near-field) when the total number of subcarriers is $M=3$, $f_c=300$ GHz and $B=30$ GHz (i.e., $f_\mathrm{Low} = 285$ GHz and $f_\mathrm{High} = 315$ GHz). In the far-field model, we can see that the array gain is maximized at the longest distance in the simulations with a deviation in the direction only. However, for the near-field case,  the array gain gives a peak at different locations (i.e., different direction and range) as the subcarrier frequency changes due to NB. This observations clearly shows that NB can lead to significant performance loses due to location/direction mismatch between the physical and spatial domains. Thus, the effect of NB should be compensated as discussed in the following experiments.

\begin{figure}[t]
	\centering		{\includegraphics[draft=false,width=.5\columnwidth]{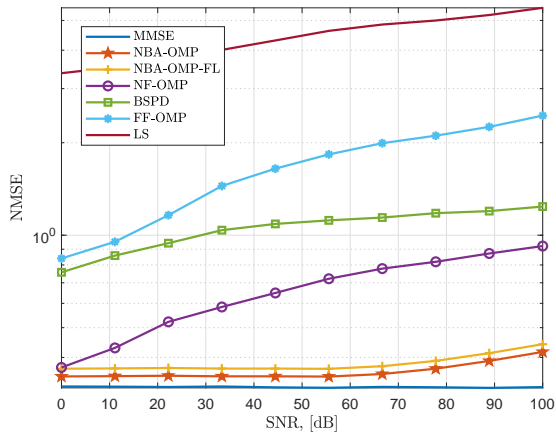}}  
	\caption{Near-field THz wideband channel estimation NMSE versus bandwidth when  $\mathrm{SNR}=10 $ dB.
	}
	\label{fig_NMSE_BW}
\end{figure}

Fig.~\ref{fig_NMSE_SNR} shows the channel estimation NMSE performance under various  SNR levels. We can see that the proposed NBA-OMP and NBA-OMP-FL techniques outperform the competing methods and closely follows the MMSE performance. The superior performance of NBA-OMP can be attributed to accurately compensating the effect of beam-split over both direction and range parameters via NBA dictionary, which is composed of SD steering vectors. On the other hand, the remaining methods fail to exhibit such high precision in high SNR regime. This is because these methods either fail to take into account the NB~\cite{nf_OMP_Dai_Wei2021Nov} or only consider the far-field signal model~\cite{alkhateeb2016frequencySelective,thz_channelEst_beamsplitPatternDetection_L_Dai}. Furthermore, the proposed approach does not require an additional hardware components in order to realize SD dictionary matrices steering vectors, hence it is hardware-efficient. Compared to NBA-OMP, the FL-based approach has slight performance loss due to model training with unevenly distributed datasets. As a result, the performance of NBA-OMP behaves like a yardstick for the FL approach since the learning-based methods cannot perform better than their labels~\cite{elbir_FL_PHY_Elbir2021Nov}.

Fig.~\ref{fig_NMSE_BW} compares the NMSE performance with respect to the bandwidth $B\in [0,100]$ GHz. We can see that the proposed NBA-OMP approaches effectively compensate the impact of beam-split for a large portion of the bandwidth up to $B<70$ GHz. 

Finally, we analyze the communication overhead of the proposed FL approach. To this end, we first compute the overhead for CL as $\mathcal{T}_\mathrm{CL} = 8\cdot\textsf{D}_k\cdot (3\cdot 8 ) = 24.57\times 10^{9}	$ when $\xi = 0$ (note that $\mathcal{T}_\mathrm{CL} = 6.31\times 10^{12}$ when $\xi = 1$), where the number of input-output tuples in the dataset is $\textsf{D}_k = 3M VG = 3\cdot 128\cdot 1000\cdot 1000 = 384\times 10^{6} $. On the other hand, the communication overhead for FL is $\mathcal{T}_\mathrm{FL} = 2\cdot 1,196,928\cdot 100\cdot 8 = 1.91\times 10^{9}$, which exhibits approximately $12$ ($3300$) times lower than that of CL when $\xi = 0$ ($\xi = 1$). This observation shows the effectiveness of the FL approach for channel estimation tasks, especially when the number of antennas is very high. It is also worthwhile mentioning that the gain achieved by employing FL is huge if the  labeling is handled at the server, i.e., $\xi = 0$.	Based on this analysis, Fig.~\ref{fig_NMSE_Training} demonstrates the channel estimation NMSE and communication-efficiency ratio, i.e., $\mathcal{R} = \frac{\mathcal{T}_\mathrm{CL}}{\mathcal{T}_\mathrm{FL}}$ with respect to the dataset size $\textsf{D} = \sum_{k = 1}^{K} \textsf{D}_k$ when $\mathrm{SNR}= 10 $ dB. We can see that smaller datasets can be used for training the learning model for faster convergence at the cost of low NMSE performance. A satisfactory channel estimation performance can be achieved (i.e., $\mathrm{NMSE} \approx 10^{-1}$) if $\textsf{D} \approx 3\times 10^{8}$.  As expected, larger datasets result in longer training times as we observed from the simulations that the number of iterations is approximately $T= 100$ for such settings. In this scenario, the communication-efficiency ratio $\mathcal{R}$ is about 12 and 3300 when $\xi = 0$ and $\xi =1$, respectively.

\begin{figure}[t]
	\centering		\subfloat[]{\includegraphics[draft=false,width=.5\columnwidth]{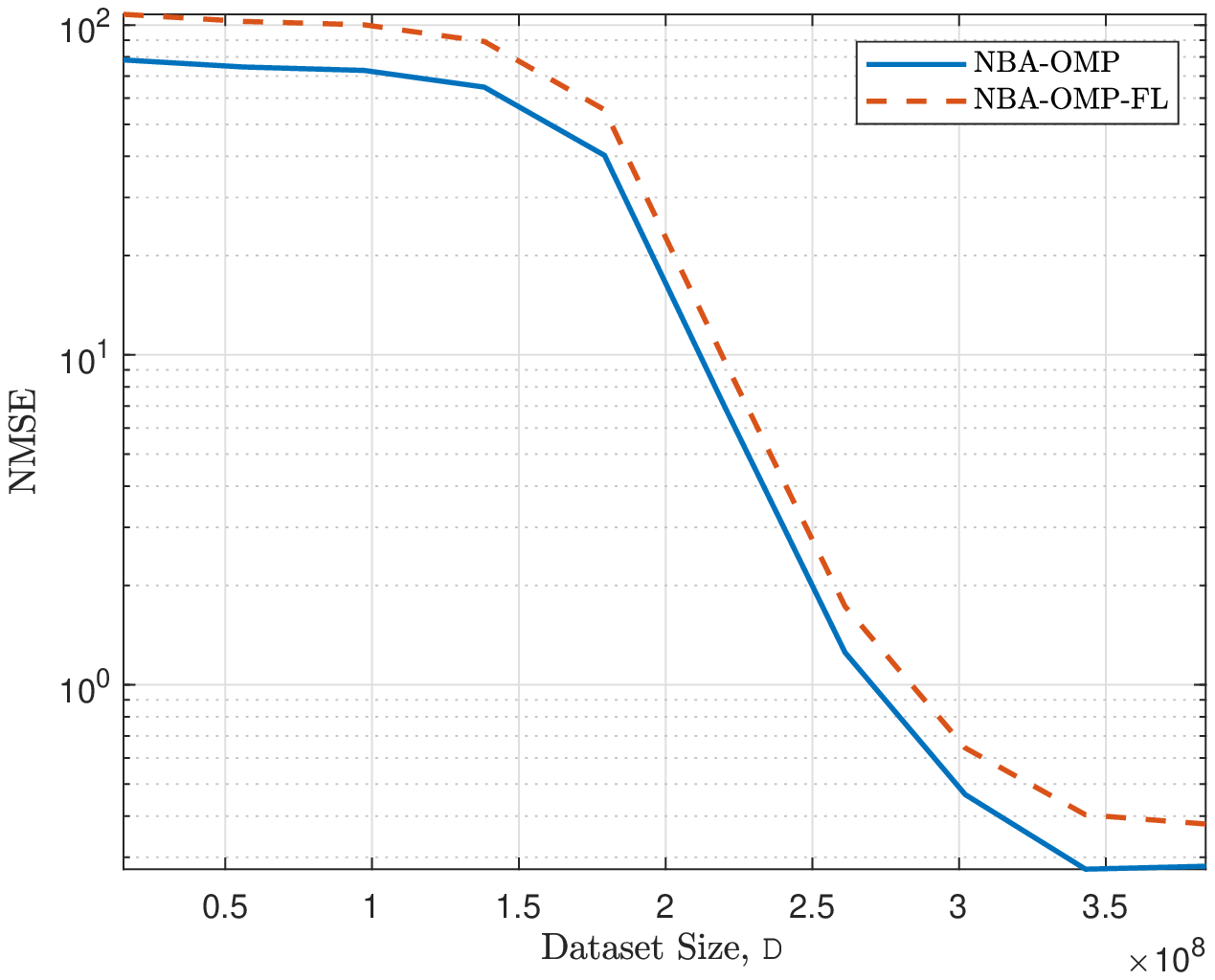}} 
	\subfloat[]{\includegraphics[draft=false,width=.5\columnwidth]{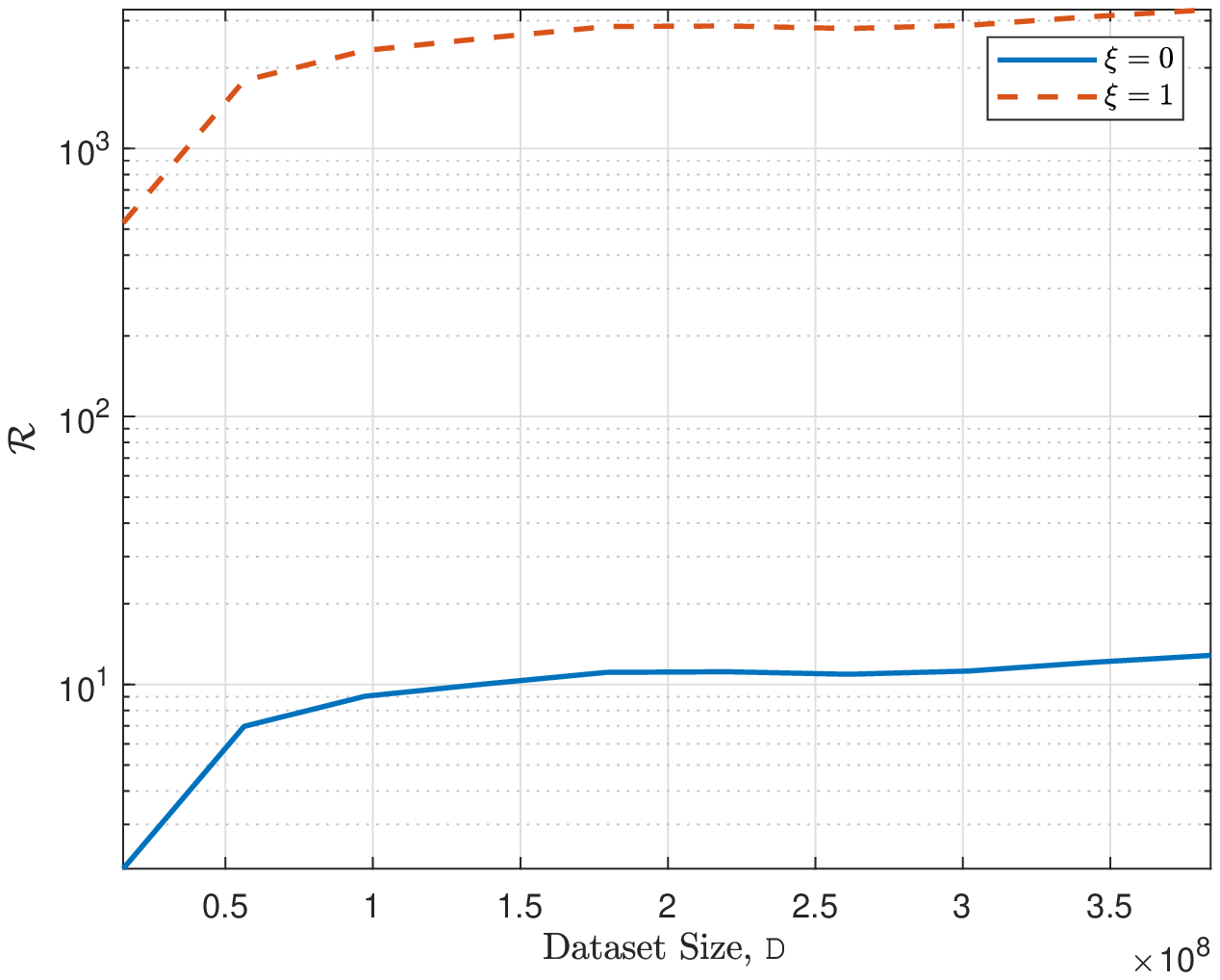}}
	\caption{Channel estimation NMSE (a) and communication-efficiency ratio (b) when  $\mathrm{SNR}=10 $ dB.
	}
	\label{fig_NMSE_Training}
\end{figure}

\section{Conclusions}
\label{sec:Conc}
We considered the near-field wideband THz channel estimation problem in the presence of near-field beam-split. We introduced both a model-based and a model-free approaches to effectively estimate the channel with lower complexity and overhead. The model-based approach is based on the OMP technique, for which an NBA dictionary is designed such that the SD DoA and range parameters are accurately matched, hence, the beam-split-corrected channel estimate is obtained. The model-free approach is based on FL scheme such that the data labels are obtained from the model-based approach. The proposed approach has lower channel and model training overhead as compared to existing techniques. Specifically, the proposed model-based approach achieves close-to-MMSE performance with $8$ times less channel usage than the conventional techniques whereas the proposed model-free technique enjoys $12$ times less communication overhead as compared to the centralized schemes.  The proposed NBA-OMP is particularly useful for THz-band applications, wherein ultra-massive number of antennas are used leading to high computational complexity, and the near-field beam-split may cause degradations in the system performance.

\balance
\bibliographystyle{IEEEtran}
\bibliography{IEEEabrv,references_119}

\end{document}